\documentclass[a4paper,english]{lipics-v2018}

\usepackage{microtype}

\usepackage{algorithm}
\usepackage{algorithmic}
\usepackage{enumitem}

\newcommand{\E}{\text{E}}
\newcommand{\diam}{\text{diam}}
\newcommand{\VLow}{V_\text{low}}
\newcommand{\VHigh}{V_\text{high}}
\newcommand{\KTZero}{$\text{KT}_0$}
\newcommand{\KTOne}{$\text{KT}_1$}

\def\polylog{\operatorname{polylog}}
\def\poly{\operatorname{poly}}

\title{Time-Message Trade-Offs in Distributed Algorithms}

\author{Robert Gmyr}{Department of Computer Science, University of Houston, Houston, USA}{rgmyr@uh.edu}{}{}

\author{Gopal Pandurangan}{Department of Computer Science, University of Houston, Houston, USA}{gopalpandurangan@gmail.com}{}{}

\authorrunning{R. Gmyr and G. Pandurangan}

\Copyright{Robert Gmyr and Gopal Pandurangan}

\subjclass{\ccsdesc[500]{Theory of computation~Distributed algorithms}}

\keywords{Randomized Algorithm, \KTOne{}  Model, Sparsifier, MST, Singular Optimality}

\category{}

\relatedversion{}

\supplement{}

\funding{Supported, in part, by NSF awards CCF-1527867, CCF-1540512,  IIS-1633720,  CCF-1717075, and BSF award 2016419.}

\EventEditors{Ulrich Schmid and Josef Widder}
\EventNoEds{2}
\EventLongTitle{32nd International Symposium on Distributed Computing (DISC 2018)}
\EventShortTitle{DISC 2018}
\EventAcronym{DISC}
\EventYear{2018}
\EventDate{October 15--19, 2018}
\EventLocation{New Orleans, USA}
\EventLogo{}
\SeriesVolume{121}
\ArticleNo{32}
\nolinenumbers
\hideLIPIcs

\begin{document}
\maketitle

\begin{abstract}
This paper focuses on showing time-message trade-offs in distributed algorithms for fundamental problems such as leader election, broadcast, spanning tree (ST), minimum spanning tree (MST), minimum cut, and many graph verification problems.
We consider  the  synchronous CONGEST distributed computing model and assume that each node has initial knowledge of itself and the identifiers of its neighbors --- the so-called {\em \KTOne{} model} --- a well-studied model that also naturally arises in many applications.
Recently, it has been established that one can obtain (almost) {\em singularly optimal} algorithms, i.e., algorithms that have {\em simultaneously} optimal time and message complexity (up to polylogarithmic factors), for many fundamental problems in the standard \KTZero{} model (where nodes have only local knowledge of themselves and not their neighbors).
The situation is less clear in the \KTOne{} model.
In this paper, we present several new distributed algorithms in the \KTOne{} model that trade off between time and message complexity.

Our distributed algorithms are based on a uniform and general approach which involves constructing a {\em sparsified spanning subgraph} of the original graph --- called a {\em danner}  ---
that trades off the {\em number of  edges} with the {\em diameter} of the sparsifier.
In particular, a key ingredient of our approach is a distributed randomized algorithm that, given a graph $G$ and any $\delta \in [0,1]$, with high probability\footnote{Throughout, by ``with high probability (w.h.p.)'' we mean with probability at least $1 - 1 / n^c$ where $n$ is the network size and $c$ is some constant.}
constructs a danner that has diameter $\tilde{O}(D + n^{1-\delta})$ and $\tilde{O}(\min \{m,n^{1+\delta} \})$ edges in $\tilde{O}(n^{1-\delta})$ rounds while using $\tilde{O}(\min\{m,n^{1+\delta}\})$ messages, where $n$, $m$, and $D$ are the number of nodes, edges, and the diameter of $G$, respectively.%
\footnote{The notation $\tilde{O}$ hides a $\text{polylog}(n)$ factor.}
Using our danner construction, we present a family of distributed randomized algorithms for various fundamental problems that exhibit a trade-off between message and time complexity and that improve over previous results.
Specifically, we show the following results (all hold with high probability) in the \KTOne{} model, which subsume and improve over prior bounds in the \KTOne{} model (King et al., PODC 2014 and Awerbuch et al., JACM 1990) and the \KTZero{} model (Kutten et al., JACM 2015, Pandurangan et al., STOC 2017 and Elkin, PODC 2017):
\begin{enumerate}
 \item {\bf Leader Election, Broadcast, and ST.}
  These problems can be solved in $\tilde{O}(D+n^{1-\delta})$ rounds using $\tilde{O}(\min\{m,n^{1+\delta}\})$ messages for any $\delta \in [0,1]$.
  \item {\bf  MST and Connectivity.}
  These problems can be solved in $\tilde{O}(D+n^{1-\delta})$ rounds using $\tilde{O}(\min\{m,n^{1+\delta}\})$ messages for any $\delta \in [0,0.5]$. In particular, for $\delta = 0.5$ we obtain a distributed MST algorithm that runs in optimal $\tilde{O}(D+\sqrt{n})$ rounds and uses $\tilde{O}(\min\{m,n^{3/2}\})$ messages. We note that this improves over the singularly optimal algorithm in the \KTZero{} model that uses $\tilde{O}(D+\sqrt{n})$ rounds and $\tilde{O}(m)$ messages.
  \item {\bf  Minimum Cut.}
  $O(\log n)$-approximate minimum cut can be solved in $\tilde{O}(D+n^{1-\delta})$ rounds using $\tilde{O}(\min\{m,n^{1+\delta}\})$ messages for any $\delta \in [0,0.5]$.
  \item {\bf Graph Verification Problems such as Bipartiteness,  Spanning Subgraph etc.}
  These can be solved in $\tilde{O}(D+n^{1-\delta})$ rounds using $\tilde{O}(\min\{m,n^{1+\delta}\})$ messages for any $\delta \in [0,0.5]$.
\end{enumerate}
\end{abstract}

\section{Introduction} \label{sec:intro}

This paper focuses on a fundamental aspect of distributed  algorithms:
{\em trade-offs} for the two basic complexity measures of {\em time} and
{\em messages}. The efficiency of distributed algorithms is traditionally measured
by their time and message (or communication) complexities.
Both complexity measures crucially influence the performance
of a distributed algorithm. Time complexity measures the number of
distributed ``rounds'' taken by the algorithm and determines
the running time of the algorithm.  Obviously, it is of interest to keep
the running time as small as possible.
Message complexity, on the other hand, measures the total amount of messages
sent and received by all the processors during the course of the
algorithm. In many applications, this is the dominant cost that also plays
a major role in determining the running time and additional costs
(e.g., energy) expended by the algorithm.
For example,  communication cost is one of the dominant costs
in the distributed computation of large-scale data~\cite{soda15}.
In another example,  in certain types of communication networks, such as ad-hoc
wireless networks, energy is a critical factor for measuring
the efficiency of a distributed algorithm \cite{khan-tpds2,choi-journal}. Transmitting a message between two
nodes in such a network has an associated energy cost and hence the total message complexity plays an important role in determining the energy
cost of an algorithm. At the same time, keeping the number of rounds small also helps in reducing the energy cost.
Thus, in various modern and emerging applications such as
resource-constrained communication networks and distributed computation
on large-scale data, it is crucial to design distributed algorithms
that optimize  both measures {\em simultaneously} \cite{podc15,soda15,spaa16,spaa18}.

Unfortunately, designing algorithms that are simultaneously time- and message-efficient has proven to be a challenging  task,
which (for some problems) stubbornly defied all prior attempts of attack.
Consequently, research in the last three decades in the area of distributed
algorithms has focused mainly on optimizing either one of the two measures
separately, typically at the cost of neglecting the other.
In this paper, we  focus  on studying the two
cost measures of time and messages \emph{jointly}, and exploring ways to design
distributed algorithms that work well under both measures
(to the extent possible).
Towards this goal, it is important to understand the
relationship between these two measures. In particular, as defined in \cite{Pandurangan0S17}, we should be able to
determine, for specific problems, whether it is possible to devise
a distributed algorithm that is either {\em singularly optimal} or  exhibits a {\em time-message trade-off}:
\begin{itemize}
\item {\bf Singularly optimal:} A distributed algorithm that is optimal with respect to both measures simultaneously
--- in which case we say that the problem enjoys {\em singular optimality}.
\item {\bf Time-message trade-off:}  Whether the problem inherently fails to admit  a  singularly optimal solution, namely,
algorithms of better time complexity for it will necessarily incur higher
message complexity and vice versa --- in which case we say that the problem exhibits a {\em time-message trade-off}.
\end{itemize}
We note that, more generally, finding a simultaneously  optimal (or almost optimal) solution may sometimes be
difficult even for ``singular optimality'' problems and hence it might be useful to first design algorithms that
have trade-offs.

This paper focuses on showing time-message trade-offs in distributed algorithms for fundamental problems such
 as leader election, broadcast, spanning tree (ST), minimum spanning tree (MST), minimum cut, and many graph verification problems. Throughout, we consider  the   synchronous CONGEST  model (see Section \ref{sec:model} for details),  a standard model in distributed computing where computation proceeds in discrete (synchronous) {\em rounds} and
 in each round only $O(\log n)$ bits are allowed to be exchanged per edge (CONGEST) where $n$ is the number of nodes in the network.  It turns out that message complexity of a distributed algorithm
 depends crucially  (as explained below) on the initial knowledge of the nodes; in this respect, there are two well-studied models --- the \KTZero{} and the \KTOne{} model.\footnote{On the other hand, for time complexity it does not really matter whether nodes have initial knowledge of just themselves (\KTZero{}) or also of their neighbors (\KTOne{}); this is because this information (i.e., the identifiers of the neighbors) can be exchanged in one round in the CONGEST model. Hence, when focusing solely on optimizing time complexity, which is the typically the case in the literature, the distinction between \KTZero{} and \KTOne{} is not important and the actual model is not even explicitly specified.}
  In the {\bf \KTZero{}} model (i.e., {\bf K}nowledge {\bf T}ill radius {\bf 0}), also called the {\em clean network model} \cite{peleg},  where nodes have initial local knowledge of only themselves (and not their neighbors),  it has (only) been recently established that one can obtain (almost) {\em singularly optimal} algorithms, i.e., algorithms that have {\em simultaneously} optimal time and message complexity (up to polylogarithmic factors), for many fundamental problems such as leader election, broadcast, ST, MST, minimum cut, and approximate shortest paths (under some conditions) \cite{JACM15,Pandurangan0S17,elkin17,Haeupler18}. More precisely, for problems such as leader election, broadcast, and ST, it has been shown \cite{JACM15} that one can design a singularly optimal algorithm in the \KTZero{} model, that takes $\tilde{O}(m)$ messages and $O(D)$ rounds --- both are tight (up to a $\text{polylog}(n)$ factor); this is
 because $\Omega(m)$ and $\Omega(D)$ are, respectively,  lower bounds on the message and time complexity for these problems in the \KTZero{} model \cite{JACM15} --- note that these lower bounds hold even for randomized Monte Carlo algorithms. The work of \cite{Pandurangan0S17}  (also see \cite{elkin17}) showed that MST is also (almost) singularly optimal, by giving a (randomized Las Vegas) algorithm that takes
 $\tilde{O}(m)$ messages and $\tilde{O}(D+\sqrt{n})$ rounds (both bounds are tight up to polylogarithmic factors). It can be shown that the singular optimality of MST in the \KTZero{} model  also implies the singular optimality of many other problems such as approximate minimum cut and graph verification problems (see Section \ref{sec:overview}).
 Recently, it was shown that approximate shortest paths and several other problems also admit singular optimality in the \KTZero{} model \cite{Haeupler18}.

 On the other hand, in the {\bf \KTOne{}} model (i.e., {\bf K}nowledge {\bf T}ill radius {\bf 1}), in which each node has initial knowledge of itself and the  {\em identifiers}\footnote{Note that only knowledge of the identifiers of neighbors is assumed, not other information such as the degree of the neighbors.} of its neighbors, the situation is less clear. The \KTOne{} model arises naturally in many settings, e.g., in networks where nodes know the identifiers
 of their neighbors (as well as other nodes), e.g., in the Internet where a node knows the IP addresses of other nodes~\cite{gopal-survey}.
 Similarly in models such as the $k$-machine model  (as well as the congested clique), it is natural to assume that each processor knows
 the identifiers of all other processors \cite{soda15,podc15,spaa16,spaa18}.
 For the \KTOne{} model, King et al.~\cite{KingKT15} showed a surprising and elegant result: There is a randomized Monte Carlo algorithm to construct an MST in $\tilde{O}(n)$ messages ($\Omega(n)$ is a message lower bound) and in $\tilde{O}(n)$ time (see Section \ref{sec:KKT}).
 Thus it is also possible to construct an ST, do leader election, and broadcast within these bounds.
 This algorithm is randomized and {\em not  comparison-based}.\footnote{Awerbuch et al. \cite{vainish} show that $\Omega(m)$ is a message lower bound for MST even in the \KTOne{} model,
if one allows only (possibly randomized Monte Carlo) {\em comparison-based} algorithms, i.e., algorithms that can operate on identifiers only by comparing them. The result of King et al.~\cite{KingKT15} breaches the $\Omega(m)$ lower bound by using non-comparison-based technique by using the identifiers as input to hash functions. Our results also breach the $\Omega(m)$ lower bound, since we use the results of King et al.\ as subroutines in our algorithms.}
 While this algorithm shows that one can achieve $o(m)$ message complexity (when $m = \omega(n \polylog n)$), it is {\em not} time-optimal --- it can take significantly more than $\tilde \Theta(D+\sqrt{n})$ rounds, which is a time lower bound even for Monte-Carlo randomized algorithms~\cite{stoc11}. In subsequent work, Mashreghi and King~\cite{MashreghiK17}  presented a trade-off between messages and time for MST: a Monte-Carlo algorithm that takes $\tilde{O}(n^{1+\epsilon}{\epsilon})$ messages and runs in $O(n/\epsilon)$ time for any $1 > \epsilon \geq \log \log n/\log n$. We note that this algorithm takes at least $O(n)$ time.

{\em A central motivating theme underlying this work is understanding the status of various fundamental problems in the \KTOne{} model --- whether they are singularly optimal or exhibit trade-offs (and, if so, to quantify the trade-offs).}
In particular, it is an open question whether one can design a randomized  (non-comparison based) MST algorithm that takes $\tilde{O}(D + \sqrt{n})$ time and $\tilde{O}(n)$ messages in the $KT_1$ model --- this would show that MST is  (almost) singularly optimal in the \KTOne{} model as well. In fact, King et al \cite{KingKT15} ask whether it is possible to construct (even) an ST in $o(n)$ rounds with $o(m)$ messages. Moreover, can we take advantage of
 the \KTOne{} model to get improved message bounds (while keeping time as small as possible) in comparison to the \KTZero{} model?

\subsection{Our Contributions and Comparison with Related Work} \label{sec:results}

In this paper, we present several results that show trade-offs between time and messages in the \KTOne{} model with respect to various problems. As a byproduct, we improve and subsume the results of \cite{KingKT15} (as well as of Awerbuch et al.~\cite{vainish}) and answer the question raised by King et al.\ on ST/MST construction at the end of the previous paragraph in the affirmative. We also show
that our results give improved bounds compared to  the results in the \KTZero{} model, including for the fundamental distributed MST problem.

 Our time-message trade-off results are based on a uniform and general approach which involves constructing
 a {\em sparsified spanning subgraph} of the original graph --- called a {\em danner} (i.e., ``diameter-preserving spanner'') --- that trades off the {\em number of  edges} with the {\em diameter} of the sparsifier (we refer to Section \ref{sec:definitions} for
   a precise definition).
{\em In particular, a key ingredient of our approach is a distributed randomized algorithm that, given a graph $G$ and any $\delta \in [0,1]$, with high probability%
\footnote{Throughout, by ``with high probability (w.h.p.)'' we mean with probability at least $1 - 1 / n^c$ where $n$ is the network size and $c$ is some constant.}
constructs a danner that has diameter $\tilde{O}(D + n^{1-\delta})$ and $\tilde{O}(\min \{m,n^{1+\delta} \})$ edges in $\tilde{O}(n^{1-\delta})$ rounds while using $\tilde{O}(\min\{m,n^{1+\delta}\})$ messages, where $n$, $m$, and $D$ are the number of nodes, edges, and the diameter of $G$, respectively.%
\footnote{The notation $\tilde{O}$ hides a $\text{polylog}(n)$ factor.}}
Using our danner construction,
 we present a family of distributed randomized algorithms for various fundamental problems that exhibit a trade-off between message and time complexity and that improve over previous results. Specifically, we show the following results (all hold with high probability) in the \KTOne{} model (cf.\ Section \ref{sec:app}):
 \begin{enumerate}
 \item {\bf Leader Election, Broadcast, and ST.}
 These problems can be solved in $\tilde{O}(D+n^{1-\delta})$ rounds using $\tilde{O}(\min\{m,n^{1+\delta}\})$ messages for any $\delta \in [0,1]$.
 These results
improve over
 prior bounds in the \KTOne{} model \cite{KingKT15, MashreghiK17,vainish} as well the \KTZero{} model  \cite{JACM15} --- discussed
 earlier in Section \ref{sec:intro}. In particular, while the time bounds in \cite{KingKT15,MashreghiK17} are always at least linear, our bounds can be sublinear and the desired time-message trade-off can be obtained by choosing an appropriate $\delta$.
 It is worth noting that the early work of Awerbuch et al.~\cite{vainish} showed that broadcast  can be solved by a deterministic algorithm in the $\text{KT}_\rho$ model using $O(\min\{m,n^{1+c/\rho}\})$ messages for some fixed constant $c>0$ in a model where each node has knowledge of the {\em topology} (not just identifiers) up to radius $\rho$. Clearly, our results improve over this (for the \KTOne{} model), since $\delta$ can be made arbitrarily small.
  \item {\bf MST and Connectivity.}
  These problems can be solved in $\tilde{O}(D+n^{1-\delta})$ rounds using $\tilde{O}(\min\{m,n^{1+\delta}\})$ messages for any $\delta \in [0,0.5]$.
  In addition to getting any desired trade-off (by plugging in an appropriate $\delta$),
  we can get a time optimal algorithm by choosing $\delta = 0.5$, which results in a distributed MST algorithm that runs in $\tilde{O}(D+\sqrt{n})$ rounds and uses $\tilde{O}(\min\{m,n^{3/2}\})$ messages.
  We note that when $m = \tilde{\omega}(n^{3/2})$, this improves over the recently proposed singularly optimal algorithms of \cite{Pandurangan0S17,elkin17} in the \KTZero{} model  that use $\tilde{O}(m)$ messages
 and   $\tilde{O}(D+\sqrt{n})$ rounds. It also subsumes and improves over the prior results of \cite{MashreghiK17,KingKT15} in the \KTOne{} model that take (essentially) $\tilde{O}(n)$ messages and $\tilde{O}(n)$ time.\footnote{Mashreghi and King \cite{MashreghiK17} also give an algorithm
with round complexity $\tilde{O}(\diam(\text{MST}))$ and with message complexity $\tilde{O}(n)$, where $\diam(\text{MST})$
is the diameter of the output MST which can be as large as $\Theta(n)$.}
 \item {\bf Minimum Cut.}
 An $O(\log n)$-approximation to the minimum cut  value (i.e., edge connectivity of the graph) can be obtained in $\tilde{O}(D+n^{1-\delta})$ rounds using $\tilde{O}(\min\{m,n^{1+\delta}\})$ messages for any $\delta \in [0,0.5]$.  Our result improves over the works of \cite{ghaffari-kuhn,hao}  that are (almost) time-optimal (i.e.,
 take $\tilde{O}(D+\sqrt{n})$ rounds), but not message optimal.
 In addition to getting any desired trade-off (by plugging in an appropriate $\delta$), in particular, if $\delta = 0.5$, we obtain a $\tilde{O}(\min\{m,n^{3/2}\})$ messages approximate minimum cut algorithm that runs in (near) optimal $\tilde{O}(D+\sqrt{n})$ rounds. This improves the best possible bounds (for $m = \tilde{\omega}(n^{3/2})$)  that can be obtained in the \KTZero{} model (cf.\ Section \ref{sec:overview}).
  \item {\bf Graph Verification Problems such as Bipartiteness, $s-t$ Cut, Spanning Subgraph.} These can be solved in $\tilde{O}(D+n^{1-\delta})$ rounds using $\tilde{O}(\min\{m,n^{1+\delta}\})$ messages for any $\delta \in [0,0.5]$.
 \end{enumerate}

\subsection{High-Level Overview of Approach} \label{sec:overview}

\textbf{Danner.}
A main technical contribution of this work is the notion of a {\em danner} and its efficient distributed construction that jointly focuses on both time and messages. As defined in Section \ref{sec:definitions}, a danner of a graph $G$ is a spanning subgraph $H$ of $G$ whose diameter, i.e., $\diam(H)$, is at most $\alpha(\diam(G)) + \beta$, where $\alpha \geq 1$ and $\beta \geq 0$ are some parameters. The goal is to construct a danner $H$ with as few edges as possible and with
$\alpha$ and $\beta$ as small as possible. It is clear that very sparse danners exist: the BFS (breadth-first spanning) tree has only $n-1$ edges and its diameter is at most twice the diameter of the graph. However, it is not clear
how to construct such a danner in a way that is efficient with respect to both messages and time, in particular in
$\tilde{O}(n)$ messages and $\tilde{O}(D)$ time, or even $o(m)$ messages and $O(D)$ time, where $D = \diam(G)$.
Note that in the \KTZero{} model, there is a tight lower bound (with almost matching upper bound) for constructing a danner: any distributed danner construction algorithm needs $\Omega(m)$ messages and $\Omega(D)$ time (this follows by reduction from leader election
which has these lower bounds \cite{JACM15} --- see Section \ref{sec:intro}).  However, in the \KTOne{} model, the status for danner construction is not known. We give (a family of) distributed algorithms for constructing a danner that trade off messages for time (Section \ref{sec:dannerConstruction}).

\noindent \textbf{Danner Construction.}
We present an algorithm (see Algorithm~\ref{alg:general}) that, given a graph $G$ and any $\delta \in [0,1]$, constructs a danner $H$ of $G$ that has $\tilde{O}(\min\{m,n^{1+\delta}\})$ edges and diameter $\tilde{O}(D+n^{1-\delta})$ (i.e., an $\tilde{O}(n^{1-\delta})$ additive danner) using $\tilde{O}(\min\{m,n^{1+\delta}\})$ messages and in $\tilde{O}(n^{1-\delta})$ time (note that the time does not depend on $D$).
The main idea behind the algorithm is as follows.
While vertices with low degree (i.e., less than $n^{\delta}$) and their incident edges can be included in the danner $H$, to handle high-degree vertices we construct a dominating set that dominates the high-degree nodes by sampling roughly $n^{1 - \delta}$ nodes (among all nodes); these are called ``center'' nodes.\footnote{
The idea of establishing a set of nodes that dominates all high-degree nodes has also been used by Aingworth et al.~\cite{AingworthCIM99} and Dor et al.~\cite{DorHZ00}.}
Each node $v$ adds the edges leading to its $\min\{\deg(v), n^\delta\}$ neighbors with the lowest identifiers (required for maintaining independence from random sampling) to $H$.
It is not difficult to argue that each high-degree node is connected to a center in $H$ and we use a relationship between the number of nodes in any dominating set and the diameter  (cf.\ Lemma \ref{lem:diameterDominationNumberBound}) to argue that the diameter of each connected component (or fragment) of $H$ is at most $\tilde{O}(n^{1-\delta})$.
We then use the FindAny algorithm of King et al.~\cite{KingKT15} to efficiently implement a distributed Boruvka-style merging (which is essentially the GHS algorithm \cite{GallagerHS83}) of fragments in the subgraph induced by the high-degree nodes and the centers.
The FindAny algorithm does not rely on identifier comparisons but instead uses random hash-functions to find an edge leaving a set of nodes very efficiently, which is crucial for our algorithm.
In each merging phase, each fragment uses FindAny to efficiently find an outgoing edge; discovered outgoing edges are added to $H$.
Only $O(\log n)$ iterations are needed to merge all fragments into a connected graph and only $\tilde{O}(\min\{m,n^{1 + \delta}\})$ messages are needed overall for merging.
At any time the set of centers in a fragment forms a dominating set of that fragment.
Thereby, the above-mentioned relationship between dominating sets and diameters guarantees that the diameters of the fragments stay within $\tilde{O}(n^{1-\delta})$.
We argue (Lemma \ref{lem:dannerfinal}) that the constructed subgraph $H$ is an additive $\tilde{O}(n^{1-\delta})$-danner of $G$.

\noindent \textbf{Danner Applications.}
What is the motivation for defining a danner and why is it useful?
The answer to both of these questions is that a danner gives a uniform way to design distributed algorithms that are both time and message efficient for various applications as demonstrated in Section~\ref{sec:app}.
Results for leader election, broadcast, and ST construction follow quite directly (cf.\ Section \ref{sec:le}):
Simply construct a danner and run the singularly optimal algorithm of \cite{JACM15} for the \KTZero{} model on the danner subgraph.
Since the danner has $\tilde{O}(n^{1+\delta})$ edges and has diameter $\tilde{O}(D+n^{1-\delta})$, this gives the required bounds.

A danner can be used to construct a MST of a graph $G$ (which also gives checking connectivity of a subgraph $H$
of $G$) using $\tilde{O}(\min\{m,n^{1+\delta}\})$ messages in
$\tilde{O}(D+n^{1-\delta})$ time, for any $\delta \in [0,0.5]$. Note that this subsumes the bounds of the singularly optimal
algorithms in the \KTZero{} model \cite{Pandurangan0S17,elkin17}.
The  distributed MST construction (cf.\ Section \ref{sec:mst}) proceeds in three steps; two of these crucially use
the danner. In the first step, we construct a danner and use it as a communication backbone to aggregate the degrees
of all nodes and thus determine $m$, the number of edges. If $m \leq n^{1+\delta}$, then we simply proceed
to use the singularly optimal algorithm of \cite{Pandurangan0S17}. Otherwise, we proceed to Step~2, where we do {\em Controlled-GHS}
which is a well-known ingredient in prior MST algorithms~\cite{dnabook,DistMst:Garay,Pandurangan0S17}. Controlled-GHS is simply Boruvka-style MST algorithm, where the diameter of all the fragments grow at a controlled rate. We use the graph sketches technique
of King et al. (in particular the FindMin algorithm --- cf.\ Section \ref{sec:KKT}) for finding outgoing edges  to keep the message complexity to $\tilde{O}(n)$. Crucially we run Controlled-GHS to only $\lceil (1-\delta)\log n\rceil$ iterations
so that the number of fragments remaining at the end of Controlled-GHS is $O(n^{\delta})$ with each having diameter
$\tilde{O}(n^{1-\delta})$; all these take only $\tilde{O}(n^{1-\delta})$ time, since the running time of Controlled-GHS
 is asymptotically bounded (up to a $\log n$ factor) by the (largest) diameter of any fragment.
 In Step~3, we merge the remaining $O(n^{\delta})$ fragments; this is done in a way that is somewhat different to
 prior MST algorithms, especially those of \cite{Pandurangan0S17,DistMst:Garay}. We simply continue the Boruvka-merging (not necessarily in a controlled way), but instead
 of using each fragment as the communication backbone, we do the merging ``globally'' using a BFS tree of the danner subgraph. The BFS tree of the danner has $O(n)$ edges and has diameter $\tilde{O}(D+n^{1-\delta})$. In each merging phase,
 each node forwards at most $O(n^{\delta})$ messages (sketches corresponding to so many distinct fragments) to the root of the BFS tree which finds the outgoing edge corresponding to each fragment (and broadcasts it to all the nodes).
 The total message cost is $\tilde{O}(n^{1+\delta})$ and, since the messages are pipelined, the total time is $\tilde{O}(D+n^{1-\delta}+n^{\delta}) = \tilde{O}(D+n^{1-\delta})$ (for $\delta = [0,0.5]$).
 Building upon this algorithm, we give time and message efficient algorithms for $O(\log n)$-approximate minimum cut and graph connectivity problems (cf.\ Section \ref{sec:cut}).

\subsection{Other Related Work} \label{sec:related}

There has been extensive research on the distributed MST problem in the \KTZero{} model, culminating in the
singularly optimal (randomized) algorithm of \cite{Pandurangan0S17} (see also \cite{elkin17}); see \cite{Pandurangan0S17} for a survey of
results in this area. The work of \cite{Pandurangan0S17} also defined the notions of {\em singular optimality} versus {\em time-message trade-offs}. Kutten et al.~\cite{JACM15} show the singular optimality of leader election (which implies the
same for broadcast and ST)  by giving tight lower bounds for both messages and time as well as giving algorithms
that simultaneously achieve the tight bounds (see Section \ref{sec:intro}).

 Compared to the \KTZero{} model, results in the \KTOne{} are somewhat less studied. The early work
of Awerbuch et al. \cite{vainish} studied time-message trade-offs for broadcast in the \KTOne{} model.
In 2015, King et al.~\cite{KingKT15} showed  surprisingly that the basic $\Omega(m)$ message lower bound that holds in the \KTZero{} model for various problems such as leader election, broadcast, MST, etc.~\cite{JACM15} can be breached in the \KTOne{} model by giving a randomized Monte Carlo algorithm to construct an MST in $\tilde{O}(n)$ messages  and in $\tilde{O}(n)$ time. The algorithm of King et al.\ uses a  powerful randomized technique of {\em graph sketches} which helps
in identifying edges going out of a cut efficiently without probing all the edges in the cut; this crucially helps in reducing
the message complexity. We heavily use this technique (as a black box) in our algorithms as well. However, note
that we could have also used other types of graph sketches (see e.g., \cite{spaa16}) which will yield similar results.

The \KTOne{} model has been assumed in other distributed computing models such as the $k$-machine model \cite{soda15,spaa16, spaa18} and the congested clique~\cite{podc15}. In \cite{podc15} it was shown
that the MST problem has a message lower bound of $\Omega(n^2)$ which can be breached in the \KTOne{} model
by using graph sketches.

Distributed minimum cut has been studied by \cite{ghaffari-kuhn,hao}, and the graph verification problems considered in this paper have been studied in \cite{stoc11}. However, the focus of these results has been on the time complexity (where \KTZero{} or \KTOne{} does not matter). We study these problems in the \KTOne{} model focusing on both time
and messages and present trade-off algorithms that also improve over the \KTZero{} algorithms (in terms of messages) --- cf.\ Section~\ref{sec:results}.
We note that $\tilde{\Omega}(D+ \sqrt{n})$ is a time lower bound for minimum cut (for any non-trivial approximation)
and for the considered graph verification problems \cite{stoc11}. It can be also shown (by using techniques
from \cite{JACM15}) that $\Omega(m)$ is
a message lower bound in the \KTZero{} model for minimum cut.

\section{Preliminaries} \label{sec:prelim}

Before we come to the main technical part of the paper, we introduce some notation and basic definitions, present our network model, and give an overview of some of the algorithms from the literature that we use for our results.

\subsection{Notation and Definitions} \label{sec:definitions}

For a graph $G$ we denote its \emph{node set} as $V(G)$ and its \emph{edge set} as $E(G)$.
For a node $u \in V(G)$ the set $N_G(u) = \{ v \in V(G) \mid \{u, v\} \in E(G) \}$ is the \emph{open neighborhood} of $u$ in $G$ and $\Gamma_G(u) = N_G(u) \cup \{ u \}$ is its \emph{closed neighborhood}.
For a set of nodes $S \subseteq V(G)$ we define $\Gamma_G(S) = \bigcup_{u \in S} \Gamma_G(u)$.
The \emph{degree} of a node $u$ in $G$ is $\deg_G(u) = |N_G(u)|$.
For a path $P = (u_0, \dots, u_\ell)$ we define $V(P)$ to be the set of nodes in $P$ and we define $|P| = \ell$ to be the \emph{length} of $P$.
A set $S \subseteq V(G)$ is a \emph{dominating set} of $G$ if $\Gamma_G(S) = V(G)$.
The \emph{domination number} $\gamma(G)$ of a graph $G$ is the size of a smallest dominating set of $G$.
The \emph{distance} $d_G(u, v)$ between two nodes $u, v \in V(G)$ is the length of a shortest path  between $u$ and $v$ in $G$.
We define the \emph{diameter} (or the {\em hop diameter}) of $G$ as $\diam(G) = \max_{u, v \in V(G)} d(u, v)$, where
the distances are taken in the graph by ignoring edge weights (i.e., each edge has weight 1).
For all of this notation, we omit $G$ when it is apparent from context.
For $S \subseteq V(G)$ the \emph{induced subgraph} $G[S]$ is defined by $V(G[S]) = S$ and $E(G[S]) = \{ \{u, v \} \in E(G) \mid u, v \in S \}.$
A subgraph $H \subseteq G$ is an \emph{$(\alpha, \beta)$-spanner} of $G$ if $V(H) = V(G)$ and $d_H(u, v) \le \alpha \cdot d_G(u, v) + \beta$ for all $u, v \in V(G)$.
In this work we make use of the weaker concept of a \emph{diameter-preserving spanner}, or in short, \emph{danner}:
A subgraph $H \subseteq G$ is a \emph{$(\alpha, \beta)$-danner} of $G$ if $V(H) = V(G)$ and $\diam(H) \le \alpha \cdot \diam(G) + \beta$.
We say $H$ is an additive $\beta$-danner if it is a $(1, \beta)$-danner.
An $(\alpha, \beta)$-spanner is also an $(\alpha, \beta)$-danner but the reverse is not generally true.
Hence, the notion of a danner is weaker than that of a spanner.

\subsection{Model} \label{sec:model}

We briefly describe the distributed computing model used.
This is the synchronous CONGEST model (see, e.g., \cite{dnabook,peleg}), which is now standard in the distributed computing literature.

A point-to-point communication network is modeled as an undirected weighted graph $G=(V,E,w)$,
where the vertices of $V$ represent the processors, the edges of $E$ represent the communication
links between them, and $w(e)$ is the weight of edge $e \in E$. Let $n = |V(G)|$ and $m = |E(G)|$. Without loss of generality,
we assume that $G$ is connected. $D$ denotes the hop-diameter (that is, the unweighted
diameter) of $G$, and, in this paper, by diameter we always mean hop-diameter.
Each node hosts a processor with limited initial knowledge. Specifically, we make the common
assumption that each node has a unique identifier (from $\{ 1, \dots, \text{poly}(n) \}$), and at the beginning of computation each
vertex $v$ accepts as input its own identifier and the weights (if any)
of the edges incident to it as well as the {\em identifiers of all its neighbors}. Thus, a node has  {\em local}
knowledge of only itself and its neighbor's identifiers; this is called the {\em \KTOne{} model}. Since each node
knows the identifier of the node on the other side of an incident edge, both endpoints can define a {\em common} edge identifier as the concatenation of identifiers of its endpoints, lowest identifier first.

The vertices are allowed to communicate through the edges of the graph $G$. It is assumed
that communication is synchronous and occurs in discrete rounds
(time steps).
In each time step, each node $v$ can send an arbitrary message of
$O(\log n)$ bits through each edge $e = \{v,u\}$ incident to $v$,
and each message arrives at $u$ by the end of this time step.
The  weights of the edges are at most
polynomial in the number of vertices $n$, and therefore the weight
of a single edge can be communicated in one time step. This model of
distributed computation is called the CONGEST$(\log n)$
model or simply the CONGEST model \cite{dnabook,peleg}. We also assume
that each vertex has access to the outcome of unbiased private coin flips.
We assume that all nodes know $n$.

\subsection{Underlying Algorithms} \label{sec:KKT}

We use an algorithm called \emph{TestOut} that was introduced by King et al.~\cite{KingKT15} in the context of computing MSTs in the \KTOne{} model.
Consider a tree $T$ that is a subgraph of a graph $G$.
The algorithm TestOut allows the nodes in $T$ to determine whether there exists an \emph{outgoing} edge, i.e., an edge that connects a node in $V(T)$ with a node in $V(G) \setminus V(T)$.
We also refer to an outgoing edge as an edge \emph{leaving} $T$.
Let $u$ be a node in $T$ that initiates an execution of TestOut.
On a high level, TestOut simply performs a single \emph{broadcast-and-echo} operation:
First, the node $u$ broadcasts a random hash function along $T$.
Each node in $T$ computes a single bit of information based on the hash function and the identifiers of the incident edges.
The parity of these bits is then aggregated using an echo (or convergecast) along $T$.
The parity is $1$ with constant probability if there is an edge in $G$ leaving $T$, and it is $0$ otherwise.
The algorithm is always correct if the parity is $1$.
The running time of the algorithm is $O(\diam(T))$ and it uses $O(|V(T)|)$ messages.

The correctness probability of TestOut can be amplified to high probability by executing the algorithm $O(\log n)$ times.
Furthermore, TestOut can be combined with a binary search on the edge identifiers to find the identifier of an outgoing edge if it exists, which adds another multiplicative factor of $O(\log n)$ to the running time and the number of messages used by the algorithm.
Finally, the procedure can also be used to find the identifier of an outgoing edge with minimum weight in a weighted graph by again using binary search on the edge weights at the cost of another multiplicative $O(\log n)$ factor.
All of these algorithms can be generalized to work on a connected subgraph $H$ that is not necessarily a tree:
A node $u \in V(H)$ first constructs a breadth-first search tree $T$ in $H$ and then executes one of the algorithms described above on $T$.
We have the following theorems.

\begin{theorem} \label{thm:findAny}
  Consider a connected subgraph $H$ of a graph $G$.
  There is an algorithm \emph{FindAny} that outputs the identifier of an arbitrary edge in $G$ leaving $H$ if there is such an edge and that outputs $\emptyset$ otherwise, w.h.p.
  The running time of the algorithm is $\tilde{O}(\diam(H))$ and it uses $\tilde{O}(|E(H)|)$ messages.
\end{theorem}

\begin{theorem} \label{thm:findMin}
  Consider a connected subgraph $H$ of a weighted graph $G$ with edge weights from $\{ 1, \dots, \poly(n) \}$.
  There is an algorithm \emph{FindMin} that outputs the identifier of a lightest edge in $G$ leaving $H$ if there is such an edge and that outputs $\emptyset$ otherwise, w.h.p.
  The running time of the algorithm is $\tilde{O}(\diam(H))$ and it uses $\tilde{O}(|E(H)|)$ messages.
\end{theorem}

We also require an efficient leader election algorithm.
The following theorem is a reformulation of Corollary 4.2 in~\cite{JACM15}.

\begin{theorem} \label{thm:LE}
  There is an algorithm that for any graph $G$ elects a leader in $O(\diam(G))$ rounds while using $\tilde{O}(|E(G)|)$ messages, w.h.p.
\end{theorem}

\section{Distributed Danner Construction} \label{sec:dannerConstruction}

The distributed danner construction presented in Algorithm~\ref{alg:general} uses a parameter $\delta$ that controls a trade-off between the time and message complexity of the algorithm.
At the same time the parameter controls a trade-off between the diameter and the size (i.e., the number of edges) of the resulting danner.
For Algorithm~\ref{alg:general} we assume that $\delta \in [0, 1)$.
We explicitly treat the case $\delta = 1$ later on.
We say a node $u$ has \emph{high degree} if $\deg(u) \ge n^\delta$.
Otherwise, it has \emph{low degree}.
Let $\VHigh$ and $\VLow$ be the set of high-degree and low-degree nodes, respectively.

\begin{algorithm}
  The algorithm constructs a danner $H$.
  Initially, we set $V(H) \gets V(G)$ and $E(H) \gets \emptyset$.
  \begin{enumerate}
    \item \label{stp:centerSelection}
    Each node becomes a \emph{center} with probability $p = \min \{ (c \log n) / n^{\delta}, 1 \}$, where $c \ge 1$ is a constant determined in the analysis.
    Let $C$ be the set of centers.
    \item \label{stp:centerConnection}
    Each node $v$  adds the edges leading to its $\min\{\deg(v), n^\delta\}$ neighbors with the lowest identifiers to $H$.
    \item \label{stp:lowDegreeBroadcast}
    Each low-degree node sends a message to all its neighbors to inform them about its low degree and whether it is a center.
    The remaining steps operate on the induced subgraphs $\hat G \gets G[\VHigh \cup C]$ and $\hat H \gets H[\VHigh \cup C]$.
    Note that every node can deduce which of its neighbors lie in $\VHigh \cup C$ from the messages sent by the low-degree nodes.
    \item \label{stp:KKT}
    For $i = 1$ to $\log n$ do the following in parallel in each connected component $K$ of $\hat H$.
    \begin{enumerate}
      \item \label{stp:LE}
      Elect a leader using the algorithm from Theorem~\ref{thm:LE}.
      \item \label{stp:findAny}
      Use the algorithm FindAny from Theorem~\ref{thm:findAny} to find an edge in $\hat G$ leaving $K$.
      If such an edge exists, add it to $H$ and $\hat H$.
      \item \label{stp:wait}
      Wait until $T$ rounds have passed in this iteration before starting the next iteration in order to synchronize the execution between the connected components.
      The value of $T$ is determined in the analysis.
    \end{enumerate}
  \end{enumerate}
  \caption{Distributed Danner Construction}
  \label{alg:general}
\end{algorithm}

We now turn to the analysis of Algorithm~\ref{alg:general}.
We assume that the probability $p$ defined in Step~\ref{stp:centerSelection} is such that $p < 1$ since for $p = 1$ the analysis becomes trivial.
Our first goal is to bound the diameter of any connected component $K$ of $\hat H$ (defined in Step~\ref{stp:lowDegreeBroadcast} of Algorithm~\ref{alg:general}) during any iteration of the loop in Step~\ref{stp:KKT}.
To achieve this goal, we first show two fundamental lemmas that allow us to bound the diameter of $K$ in terms of its domination number $\gamma(K)$ (see Section~\ref{sec:definitions}).
The main observation behind Lemma~\ref{lem:shortestPathNeighbors} was also used by Feige et al.\ in~\cite{FeigePRU90}.

\begin{lemma} \label{lem:shortestPathNeighbors}
  Let $P$ be a shortest path in a graph $G$.
  For each node $v \in V(G)$ it holds $|\Gamma(v) \cap V(P)| \le 3$.
\end{lemma}
\begin{proof}
  We show the lemma by contradiction.
  Let $P = (u_0, \dots, u_\ell)$ be a shortest path in $G$.
  Suppose there is a node $v \in V(G)$ such that $|\Gamma(v) \cap V(P)| \ge 4$.
  Let $u_i$ be the node in $\Gamma(v) \cap V(P)$ with the lowest index in $P$ and let $u_j$ be the node in $\Gamma(v) \cap V(P)$ with the highest index in $P$.
  Since $|\Gamma(v) \cap V(P)| \ge 4$ at least two nodes lie between $u_i$ and $u_j$ in $P$.
  We distinguish two cases.
  If $u_i = v$ or $u_j = v$ then $P' = (u_0, \dots, u_i, u_j, \dots, u_\ell)$ is a path in $G$ such that $|P'| \le |P| - 2$, which is a contradiction.
  Otherwise, the path $P' = (u_0, \dots, u_i, v, u_j \dots, u_\ell)$ is a path in $G$ such that $|P'| \le |P| - 1$, which is again a contradiction.
\end{proof}

\begin{lemma} \label{lem:diameterDominationNumberBound}
  For a connected graph $G$ it holds $\diam(G) < 3 \gamma(G)$.
\end{lemma}
\begin{proof}
  We show the lemma by contradiction.
  Suppose there is a shortest path $P$ in $G$ such that $|P| \ge 3 \gamma(G)$.
  Let $S$ be a dominating set in $G$ with $|S| = \gamma(G)$.
  By definition, for each node $u \in V(P)$ there is a node $v \in S$ such that $u \in \Gamma(v)$.
  Since $|V(P)| = |P| + 1 > 3 |S|$, the pigeonhole principle implies that there must be a node $v \in S$ such that $|\Gamma(v) \cap V(P)| > 3$.
  By Lemma~\ref{lem:shortestPathNeighbors}, this implies that $P$ is not a shortest path, which is a contradiction.
\end{proof}

With these lemmas in place, we can now turn to the problem of bounding the diameter of a connected component $K$ of $\hat H$.
We first bound the number of centers established in Step~\ref{stp:centerSelection}.

\begin{lemma} \label{lem:numCenters}
  It holds $|C| = \tilde{O}(n^{1 - \delta})$, w.h.p.
\end{lemma}
\begin{proof}
  Let $X_u$ be a binary random variable such that $X_u = 1$ if and only if $u \in C$.
  We have $\E[X_u] = (c \log n) / n^{\delta}$.
  By definition it holds $|C| = \sum_{u \in V} X_u$.
  The linearity of expectation implies $\E[\, |C| \,] = \sum_{u \in V} E[X_u] = c n^{1 - \delta} \log n$.
  Since $|C|$ is a sum of independent binary random variables we can apply Chernoff bounds (see, e.g.,~\cite{MitzenmacherUpfal}) to get $\Pr[\, |C| \ge 2 c n^{1 - \delta} \log n \, ] \le \exp(-c n^{1 - \delta} \log n / 3)$.
  The lemma follows by choosing $c$ sufficiently large.
\end{proof}

The next two lemmas show that the set of centers in $K$ forms a dominating set of $K$.

\begin{lemma} \label{lem:adjacent}
  After Step~\ref{stp:centerConnection} each high-degree node is adjacent to a center in $\hat H$, w.h.p.
\end{lemma}
\begin{proof}
  Consider a node $u \in \VHigh$.
  Let $S$ be the set of the $n^\delta$ neighbors of $u$ with lowest identifier.
  Each node in $S$ is a center with probability $p$.
  Hence, the probability that no node in $S$ is a center is
  $
    \left(1 - p \right)^{|S|}
    = \left(1 - (c \log n) / n^{\delta} \right)^{n^\delta}
    \le \exp(-c \log n)
  $.
  The lemma follows by applying the union bound over all nodes and choosing the constant $c$ sufficiently large.
\end{proof}

\begin{lemma} \label{lem:componentDominatingSet}
  Let $K$ be a connected component of $\hat H$ before any iteration of the loop in Step~\ref{stp:KKT} or after the final iteration.
  The set of centers in $K$ is a dominating set of $K$, w.h.p.
\end{lemma}
\begin{proof}
  Recall that $V(K) \subseteq \VHigh \cup C$ by definition.
  Hence, each node $u \in V(K)$ is a center or has high degree.
  If $u$ is a center, there is nothing to show.
  If $u$ is not a center, it must be of high degree.
  According to Lemma~\ref{lem:adjacent}, $u$ is connected to a center $v$ in $\hat H$.
  This implies that $v \in V(K)$ and $\{u, v\} \in E(K)$.
\end{proof}

By combining the statements of Lemmas~\ref{lem:diameterDominationNumberBound},~\ref{lem:numCenters}~and~\ref{lem:componentDominatingSet}, we get the following lemma.

\begin{lemma} \label{lem:sumComponentDiameters}
  Let $K_1, \dots, K_r$ be the connected components of $\hat H$ before any iteration of the loop in Step~\ref{stp:KKT} or after the final iteration.
  It holds $\sum_{i = 1}^r \diam(K_i) = \tilde{O}(n^{1 - \delta})$, w.h.p.
\end{lemma}
\begin{proof}
  Let $C(K_i)$ be the set of centers in $K_i$.
  According to Lemma~\ref{lem:componentDominatingSet}, $C(K_i)$ is a dominating set of $K_i$.
  Therefore, Lemma~\ref{lem:diameterDominationNumberBound} implies $\diam(K_i) < 3 |C(K_i)|$.
  This implies $\sum_{i = 1}^r \diam(K_i) < 3 \sum_{i = 1}^r |C(K_i)| = 3 |C| = \tilde{O}(n^{1 - \delta})$, where the last equality holds according to Lemma~\ref{lem:numCenters}.
\end{proof}

The following simple corollary gives us the desired bound on the diameter of a connected component $K$ of $\hat H$.

\begin{corollary} \label{cor:componentDiameter}
  Let $K$ be a connected component of $\hat H$ before any iteration of the loop in Step~\ref{stp:KKT} or after the final iteration.
  It holds $\diam(K) = \tilde{O}(n^{1 - \delta})$, w.h.p.
\end{corollary}

On the basis of Corollary~\ref{cor:componentDiameter}, we can bound the value of $T$, the waiting time used in Step~\ref{stp:wait}:
Consider an iteration of the loop in Step~\ref{stp:KKT}.
For each connected component $K$ the leader election in Step~\ref{stp:LE} can be achieved in $\tilde{O}(n^{1 - \delta})$ rounds according to Theorem~\ref{thm:LE}.
The algorithm FindAny in Step~\ref{stp:findAny} requires $\tilde{O}(n^{1 - \delta})$ rounds according to Theorem~\ref{thm:findAny}.
Therefore, we can choose $T$ such that $T = \tilde{O}(n^{1 - \delta})$.

Our next objective is to show that the computed subgraph $H$ is an additive $\tilde{O}(n^{1 - \delta})$-danner.
To this end, we first take a closer look at the connected components of $\hat H$ after the algorithm terminates.

\begin{lemma} \label{lem:componentMaximality}
  After Algorithm~\ref{alg:general} terminates, the set of connected components of $\hat H$ equals the set of connected components of $\hat G$.
\end{lemma}
\begin{proof}
  Consider a connected component $K_{\hat G}$ of $\hat G$.
  We show by induction that after iteration $i$ of the loop in Step~\ref{stp:KKT}, each connected component of $\hat H[V(K_{\hat G})]$ has size at least $\min \{ 2^i, |V(K_{\hat G})| \}$.
  Since $|V(K_{\hat G})| \le n$ and the loop runs for $\log n$ iterations, this implies that after the algorithm terminates, only one connected components remains in $\hat H[V(K_{\hat G})]$.

  The statement clearly holds before the first iteration of the loop, i.e., for $i = 0$.
  Suppose that the statement holds for iteration $i \ge 0$.
  We show that it also holds for iteration $i + 1$.
  If there is only one connected component at the beginning of iteration $i + 1$ then that connected component must equal $K_{\hat G}$ so the statement holds.
  If there is more than one connected component at the beginning of iteration $i + 1$ then by the induction hypothesis each connected component has size at least $2^i$.
  Each connected component finds an edge leading to another connected component in Step~\ref{stp:findAny} and thereby merges with at least one other connected component.
  The size of the newly formed component is at least $\min \{ 2^{i + 1}, |K_{\hat G}| \}$.
\end{proof}

We are now ready to show that $H$ is an additive $\tilde{O}(n^{1 - \delta})$-danner.

\begin{lemma} \label{lem:dannerfinal}
  Algorithm~\ref{alg:general} computes an additive $\tilde{O}(n^{1 - \delta})$-danner $H$ of $G$, w.h.p.
\end{lemma}
\begin{proof}
  Let $P_G = (u_0, \dots, u_\ell)$ be a shortest path in $G$.
  We construct a path $P_H$ from $u_0$ to $u_\ell$ in $H$ such that $|P_H| \le |P_G| + \tilde{O}(n^{1 - \delta})$.
  Some of the edges in $P_G$ might be missing in $H$.
  Let $\{ u_i, u_{i + 1} \}$ be such an edge.
  Observe that if $u_i$ or $u_{i + 1}$ has low degree then the edge $\{ u_i, u_{i + 1} \}$ is contained in $H$ since a low degree node adds all of its incident edges to $H$ in Step~\ref{stp:centerConnection}.
  Hence, $u_i$ and $u_{i + 1}$ must have high degree.
  Since the nodes share an edge in $G$, they lie in the same connected component of $\hat G$.
  According to Lemma~\ref{lem:componentMaximality} this means that the nodes also lie in the same connected component of $\hat H$.
  Therefore, there is a path in $\hat H$ between $u_i$ and $u_{i + 1}$.
  We construct $P_H$ from $P_G$ by replacing each edge $\{ u_i, u_{i + 1} \}$ that is missing in $H$ by a shortest path from $u_i$ to $u_{i + 1}$ in $\hat H$.

  While $P_H$ is a valid path from $u_0$ to $u_\ell$ in $H$, its length does not necessarily adhere to the required bound.
  To decrease the length of $P_H$, we do the following for each connected component $K$ of $\hat H$:
  If $P_H$ contains at most one node from $K$, we proceed to the next connected component.
  Otherwise, let $v$ be the first node in $P_H$ that lies in $K$ and let $w$ be the last node in $P_H$ that lies in $K$.
  We replace the subpath from $v$ to $w$ in $P_H$ by a shortest path from $v$ to $w$ in $\hat H$.
  After iteratively applying this modification for each connected component, the path $P_H$ enters and leaves each connected component of $\hat H$ at most once and within each connected component $P_H$ only follows shortest paths.
  Hence, according to Lemma~\ref{lem:sumComponentDiameters}, the number of edges in $P_H$ passing through $\hat H$ is bounded by $\tilde{O}(n^{1 - \delta})$.
  The remaining edges in $P_H$ stem from $P_G$, so their number is bounded by $|P_G|$.
  In summary, we have $|P_H| \le |P_G| + \tilde{O}(n^{1 - \delta})$.
\end{proof}

To complete our investigation we analyze the time and message complexity of Algorithm~\ref{alg:general} and bound the number of edges in the resulting danner $H$.

\begin{lemma} \label{lem:numEdges}
  The running time of Algorithm~\ref{alg:general} is $\tilde{O}(n^{1 - \delta})$ and the number of messages sent by the algorithm is $\tilde{O}(\min\{m,n^{1 + \delta}\})$.
  After the algorithm terminates it holds $|E(H)| = \tilde{O}(\min\{m,n^{1 + \delta}\})$.
\end{lemma}
\begin{proof}
  We begin with the running time.
  The first three steps of the algorithm can be executed in a single round.
  The loop in Step~\ref{stp:KKT} runs for $\log n$ iterations, each of which takes $T = \tilde{O}(n^{1 - \delta})$ rounds.

  Next we bound the number of edges in the danner.
  Steps~\ref{stp:centerSelection}~and~\ref{stp:lowDegreeBroadcast} do not add any edges to $H$.
  Step~\ref{stp:centerConnection} adds $\tilde{O}(\min\{m,n^{1 + \delta}\})$ edges to $H$.
  The loop in Step~\ref{stp:KKT} runs for $\log n$ iterations, and in every iteration each connected component of $\hat H$ adds at most one edge to $H$.
  Since the number of connected components is at most $n$ at all times, the total number of edges added in Step~\ref{stp:KKT} is $\tilde{O}(n)$.

  Finally, we turn to the message complexity of the algorithm.
  In Step~\ref{stp:centerSelection} the nodes send no messages.
  The number of messages sent in Step~\ref{stp:centerConnection} is $\tilde{O}(\min\{m,n^{1 + \delta}\})$.
  In Step~\ref{stp:lowDegreeBroadcast} each low-degree node sends a message to each of its neighbors.
  By definition a low-degree node has at most $n^\delta$ neighbors and there are at most $n$ low-degree nodes.
  Therefore, at most $\tilde{O}(\min\{m,n^{1 + \delta}\})$ messages are sent in this step.
  Each iteration of the loop in Step~\ref{stp:KKT} operates on a subgraph $\hat H$ of the final danner $H$.
  Consider a connected component $K$ of $\hat H$.
  Both the leader election in Step~\ref{stp:LE} and the algorithm FindAny in Step~\ref{stp:findAny} use $\tilde{O}(|E(K)|)$ messages according to Theorems~\ref{thm:LE}~and~\ref{thm:findAny}, respectively.
  Hence, the overall number of messages used in any iteration is $\tilde{O}(|E(\hat H)|)$ which is bounded by $\tilde{O}(|E(H)|) = \tilde{O}(\min\{m,n^{1 + \delta}\})$.
\end{proof}

Finally, we treat the special case $\delta = 1$.
In this case we do not use Algorithm~\ref{alg:general} but instead let each node add all its incident edges to $H$ such that $H = G$.
Combining the statements of the previous two lemmas together with the special case of $\delta = 1$ yields the following theorem.

\begin{theorem} \label{thm:general}
  There is an algorithm that for a connected graph $G$ and any $\delta \in [0, 1]$ computes an additive $\tilde{O}(n^{1 - \delta})$-danner $H$ consisting of $\tilde{O}(\min\{m,n^{1 + \delta}\})$ edges, w.h.p.
  The algorithm takes $\tilde{O}(n^{1 - \delta})$ rounds and requires $\tilde{O}(\min\{m,n^{1 + \delta}\})$ messages.
\end{theorem}

\section{Applications} \label{sec:app}

In this section we demonstrate that the danner construction presented in Section~\ref{sec:dannerConstruction} can be used to establish trade-off results for many fundamental problems in distributed computing.

\subsection{Broadcast, Leader Election, and Spanning Tree} \label{sec:le}

On the basis of the danner construction presented in Section~\ref{sec:dannerConstruction} it is easy to obtain a set of trade-off results for broadcast, leader election, and spanning tree construction.
The number of messages required for a broadcast can be limited by first computing a danner and then broadcasting along the danner.
For leader election we can run the algorithm of Kutten et al.~\cite{JACM15} mentioned in Theorem~\ref{thm:LE} on the computed danner.
Finally, for spanning tree construction we can elect a leader which then performs a distributed breadth-first search on the danner to construct the spanning tree.
We have the following theorem.

\begin{theorem} \label{thm:simpleProblems}
  There are algorithms that for any connected graph $G$ and any $\delta \in [0, 1]$ solve the following problems in $\tilde{O}(D + n^{1 - \delta})$ rounds while using $\tilde{O}(\min\{m,n^{1 + \delta}\})$ messages, w.h.p.:
  broadcast, leader election, and spanning tree.
\end{theorem}

\subsection{Minimum Spanning Tree and Connectivity} \label{sec:mst}

In this section we assume that we are given a \emph{weighted} connected graph $G$ with edge weights from $\{ 1, \dots, \poly(n) \}$.
Without loss of generality we assume that the edge weights are distinct such that the MST is unique.
We present a three step algorithm for computing the MST.

\textbf{Step 1:}
We compute a spanning tree of $G$ using the algorithm described in Section~\ref{sec:le} while ignoring the edge weights.
Recall that the algorithm computes the spanning tree by having a leader node initiate a distributed breadth-first search along a danner of diameter $\tilde{O}(D + n^{1 - \delta})$.
Therefore, the algorithm supplies us with a \emph{rooted} spanning tree $T$ of depth $\tilde{O}(D + n^{1 - \delta})$.
We aggregate the number of edges $m$ in $G$ using a convergecast along $T$.
If $m \leq n^{1+\delta}$, we execute the singularly optimal algorithm of Pandurangan et al.~\cite{Pandurangan0S17} on the original graph $G$ to compute the MST, which takes $\tilde{O}(D + \sqrt{n})$ rounds and requires $\tilde{O}(m)$ messages.
Otherwise, we proceed with the following steps.

\textbf{Step 2:}
We execute the so-called Controlled-GHS procedure on $G$ as described in~\cite{Pandurangan0S17}.
This procedure is a modified version of the classical Gallager-Humblet-Spira (GHS) algorithm for distributed MST~\cite{GallagerHS83}.
It constructs a set of \emph{MST fragments} (i.e., connected subgraphs of the MST).
However, in contrast to the original GHS, Controlled-GHS limits the diameter of the fragments by controlling the way in which fragments are merged.
By running Controlled-GHS for $\lceil (1 - \delta) \log n \rceil$ iterations we get a spanning forest consisting of at most $n^\delta$ MST-fragments, each having diameter $O(n^{1-\delta})$ in $\tilde{O}(n^{1 - \delta})$ rounds.
The Controlled-GHS described in~\cite{Pandurangan0S17} requires $\tilde{O}(m)$ messages.
However, we can reduce the number of messages to $\tilde{O}(n)$ without increasing the running time by modifying the Controlled-GHS procedure to use the algorithm FindMin described in Theorem~\ref{thm:findMin} to find the lightest outgoing edge of a fragment.

\textbf{Step 3:}
Our goal in the final step of the algorithm is to merge the remaining $n^\delta$ MST fragments quickly.
This step executes the same procedure for $\log n$ iterations.
Each iteration reduces the number of fragments by at least a factor of two so that in the end only a single fragment remains, which is the MST.

We use a modified version of the algorithm TestOut that only communicates along the spanning tree $T$ computed in Step~1 (i.e., it ignores the structure of the fragments) and that operates on all remaining fragments in parallel.
Recall that the original TestOut algorithm for a single connected component consists of a broadcast-and-echo in which a leader broadcasts a random hash function and the nodes use an echo (or convergecast) to aggregate the parity of a set of bits.
We can parallelize this behavior over all fragments as follows:
Let $v_T$ be the root of $T$.
First, $v_T$ broadcasts a random hash function through $T$.
The same hash function is used for all fragments.
Each node $u$ uses the hash function to compute its individual bit as before and prepares a message consisting of the identifier of the fragment containing $u$ and the bit of $u$.
These messages are then aggregated up the tree in a pipelined fashion:
In each round, a node sends the message with the lowest fragment identifier to its parent.
Whenever a node holds multiple messages with the same fragment identifier, it combines them into a single message consisting of the same fragment identifier and the combined parity of the bits of the respective messages.
Since $T$ has depth $\tilde{O}(D + n^{1 - \delta})$ and there are at most $n^\delta$ different fragment identifiers, $v_T$ learns the aggregated parity of the bits in each individual fragment after $\tilde{O}(D + n^{1 - \delta} + n^\delta)$ rounds, which completes the parallelized execution of TestOut.

As explained in Section~\ref{sec:KKT}, a polylogarithmic number of executions of TestOut in combination with binary search can be used to identify the lightest outgoing edge of a fragment.
The ranges for the binary search for each fragment can be broadcast by $v_T$ in a pipelined fashion and the TestOut procedure can be executed in parallel for all fragments as described above.
Thereby, $v_T$ can learn the identifier of a lightest outgoing edge for each fragment in parallel.
To merge the fragments, $v_T$ does the following:
First, it learns the fragment identifiers of the nodes at the end of each outgoing edge.
It then locally computes the changes in the fragment identifiers that follow from the merges.
Finally, it broadcasts these changes along with the identifiers of the leaving edges to merge the fragments.
This completes one iteration of the procedure.

Overall, the operations of the final step can be achieved a using polylogarithmic number of pipelined broadcast-and-echo operations.
Therefore, the running time of this step is $\tilde{O}(D + n^{1 - \delta} + n^\delta)$ rounds.
In each pipelined broadcast-and-echo each node sends at most $n^\delta$ messages, so the overall number of messages is $\tilde{O}(n^{1 + \delta})$.
This gives us the following theorem.

\begin{theorem}
  There is an algorithm that for any connected graph $G$ with edge weights from $\{ 1, \dots, \poly(n) \}$ and any $\delta \in [0, 0.5]$ computes an MST of $G$ in $\tilde{O}(D + n^{1 - \delta})$ rounds while using $\tilde{O}(\min\{m,n^{1 + \delta}\})$ messages, w.h.p.
\end{theorem}

For $\delta = 0.5$ we get an algorithm with optimal running time up to polylogarithmic factors.

\begin{corollary}
  There is an algorithm that for any connected graph $G$ with edge weights from $\{ 1, \dots, \poly(n) \}$ computes an MST of $G$ in $\tilde{O}(D + \sqrt{n})$ rounds while using $\tilde{O}(\min\{m,n^{3 / 2}\})$ messages, w.h.p.
\end{corollary}

Using this result on MST, it is not hard to devise an algorithm that computes the connected components of a subgraph $H$ of $G$ (and thus also test connectivity):
We assign the weight $0$ to each edge in $E(H)$ and the weight $1$ to each edge in $E(G) \setminus E(H)$.
We then run a modified version of the above MST algorithm in which a fragment stops participating as soon as it discovers that its lightest outgoing edge has weight $1$.
Thereby, fragments only merge along edges in $H$.
Once the algorithm terminates, any two nodes in the same connected component of $H$ have the same fragment identifier while any two nodes in distinct connected components have distinct fragment identifiers.

\begin{corollary}
  There is an algorithm that for any graph $G$, any subgraph $H$ of $G$, and any $\delta \in [0, 0.5]$ identifies the connected components of $H$ in $\tilde{O}(D + n^{1 - \delta})$ rounds while using $\tilde{O}(\min\{m,n^{1 + \delta}\})$ messages, w.h.p.
\end{corollary}

\subsection{$O(\log n)$-Approximate Minimum Cut} \label{sec:cut}

We  describe an algorithm
that finds an $O(\log n)$-approximation to the edge connectivity of the graph (i.e., the minimum cut value).

\begin{theorem}
\label{thm:mincut}
There is a distributed algorithm for finding an $O(\log n)$-approximation to the edge connectivity of the graph (i.e., the minimum cut value)
that uses $\tilde{O}(\min\{m,n^{1+\delta}\})$ messages and runs in $\tilde{O}(D+n^{1-\delta})$ rounds for any $\delta \in [0,0.5]$, w.h.p.
\end{theorem}

The main idea behind the algorithm is based on the following sampling theorem.

\begin{theorem}[\cite{Karger,ghaffari-kuhn}]
\label{thm:karger}
Consider an arbitrary unweighted multigraph\footnote{Note that
a weighted graph with polynomially large edge weights can be represented as an unweighted graph with polynomial
number of multiedges.} $G=(V,E)$ with edge connectivity  $\lambda$ and choose subset $S \subseteq E$ by including each edge $e \in E$ in set $S$ independently with probability  $p$. If $p \geq c\log n/\lambda$, for a sufficiently large (but fixed) constant $c$, then
the sampled subgraph $G' = (V,S)$ is connected, w.h.p.
\end{theorem}

We next sketch  the distributed algorithm  that is claimed in Theorem \ref{thm:mincut}. The distributed algorithm implements
the above sampling theorem which provides a  simple approach for finding an $O(\log n)$-approximation of the edge connectivity of  $G$ by  sampling subgraphs with exponentially growing sampling probabilities (e.g., start with an estimate of $\lambda = m$, the total number of (multi-)edges, and keep decreasing the estimate by a factor of 2) and checking the connectivity of each sampled subgraph. We take the first value of the estimate where the sampled graph is connected as the $O(\log n)$-approximate value. This algorithm can be easily implemented in the \KTZero{} distributed model using $\tilde{O}(m)$ messages and $\tilde{O}(D+ \sqrt{n})$ rounds by using
the singularly optimal MST algorithm of \cite{Pandurangan0S17} or \cite{elkin17} (which can be directly used to test connectivity --- cf.\ Section \ref{sec:mst}). However, implementing the algorithm in the \KTOne{} model in the prescribed time and message bounds of Theorem \ref{thm:mincut} requires some care, because of implementing the sampling step; each endpoint has to agree on the sampled edge without actually communicating through that edge.

The sampling step can be accomplished by sampling with a {\em strongly $O(\log n)$-universal hash function}.
 Such a hash function
can be created by using only $O(\log n)$ independent shared random bits (see, e.g., \cite{carter}).
The main insight is that a danner allows sharing of random bits efficiently.
This can be done by constructing a danner\footnote{Note that the danner of a multigraph is constructed by treating multi-edges as a single edge.}  and letting the leader (which can be elected using the danner --- cf.\ Section \ref{sec:le}) generate $O(\log n)$ independent random bits and broadcast them to all the nodes via the danner. The nodes can then construct and
use the hash function to sample the edges. However, since the hash function is only $O(\log n)$-universal it is only $O(\log n)$-wise independent.
But still one can show that the guarantees  of Theorem \ref{thm:karger}  hold  if the edges are sampled by a $O(\log n)$-wise independent hash function. This can be seen by using Karger's proof \cite{Karger} and checking that Chernoff bounds for $O(\log n)$-wise independent random variables (as used in
\cite{aravind}) are (almost) as good as that as (fully) independent random variables. Hence edge sampling can be accomplished using time $\tilde{O}(D+n^{1-\delta})$ (the diameter of the danner) and messages $\tilde{O}(\min\{m,n^{1+\delta})$ (number of edges of the danner). Once the sampling step is done, checking connectivity can be done by using the algorithm of Section \ref{sec:mst}.

\subsection{Algorithms for Graph Verification Problems} \label{sec:verification}

It is well known that graph connectivity is an important building block for several graph verification problems
(see, e.g., \cite{stoc11}).  Thus, using the connectivity algorithm of Section \ref{sec:mst} as a subroutine, we can show that the problems stated in the theorem below  (these are formally defined, e.g.,
in Section~2.4 of \cite{stoc11}) can be solved in the \KTOne{} model (see, e.g., \cite{stoc11,spaa16}).

\begin{theorem}\label{thm:verification}
There exist distributed algorithms in the \KTOne{} model that solve  the following verification problems in
$\tilde{O}(\min\{m,n^{1+\delta}\})$ messages and $\tilde{O}(D+n^{1-\delta})$ rounds, w.h.p, for any $\delta \in [0,0.5]$
: spanning connected subgraph, cycle containment, $e$-cycle containment, cut, $s$-$t$ connectivity,
edge on all paths, $s$-$t$ cut, bipartiteness.
\end{theorem}

\section{Conclusion} \label{sec:conc}

This work is a step towards understanding time-message trade-offs for distributed algorithms in the \KTOne{} model.
Using our danner construction, we obtained algorithms that exhibit time-message trade-offs across the spectrum
by choosing the parameter $\delta$ as desired.
There are many key open questions raised by our work.
First, it is not clear whether one can do better than the algorithms we obtained here for various fundamental problems in the \KTOne{} model.
In particular, it would be interesting to know whether there are singularly optimal algorithms in the \KTOne{} model --- e.g., for leader election, can we show an $\tilde{O}(n)$ messages algorithm that runs in $\tilde{O}(D)$ (these are respective lower bounds for messages and time in \KTOne{}); and, for MST, can we show an  $\tilde{O}(n)$ messages algorithm that runs in $\tilde{O}(D+\sqrt{n})$.
A related question is whether one can construct an $(\alpha, \beta)$-danner with $\tilde{O}(n)$ edges and $\alpha, \beta = \tilde{O}(1)$ in a distributed manner using $\tilde{O}(n)$ messages and in $\tilde{O}(D)$ time.
Such a construction could be used to obtain singularly optimal algorithms for other problems. Finally, our danner construction is randomized; a deterministic construction with similar guarantees will yield deterministic algorithms.

\bibliography{bib}

\end{document}